\documentclass[11pt]{article}

\usepackage{booktabs} 

\usepackage{bbm}
\usepackage{enumitem}
\usepackage{times}
\usepackage{fullpage}

\usepackage{amsmath}
\usepackage{amsfonts}
\usepackage{amsthm}
\usepackage[utf8]{inputenc}
\usepackage[english]{babel}

\newtheorem{theorem}{Theorem}
\newtheorem{corollary}{Corollary}
\newtheorem{lemma}{Lemma}
\newtheorem{definition}{Definition}

\usepackage{authblk}
\allowdisplaybreaks 

\begin{document}

\title{Online Non-preemptive Scheduling on Unrelated Machines with Rejections}

\author[1]{Giorgio Lucarelli \thanks{giorgio.lucarelli@imag.fr}}
\author[2]{Benjamin Moseley \thanks{moseleyb@andrew.cmu.edu}}
\author[3]{Nguyen Kim Thang \thanks{thang@ibisc.fr}}
\author[4,5]{Abhinav Srivastav \thanks{abhinav.srivastav@ens.fr}}
\author[1]{Denis Trystram \thanks{trystram@imag.fr\\ \\}}

\affil[1]{Universit\'e GrenobleAlpes, CNRS, INRIA, Grenoble INP, LIG\\}
\affil[2]{Carnegie Mellon University\\}
\affil[3]{IBISC, Univ \'Evry, Universit\'e Paris-Saclay\\}
\affil[4]{D\'epartment d'Informatique, ENS-Paris\\}
\affil[5]{LAMSADE, Universit\'e Paris Dauphine}


\maketitle

\begin{abstract}
When a computer system schedules jobs there is typically a significant cost associated with preempting a job during execution. This cost can be from the expensive task of saving the memory's state and loading data into and out of memory. It is desirable to schedule jobs non-preemptively to avoid the costs of preemption.  

There is a need for non-preemptive system schedulers on desktops, servers and data centers. Despite this need, there is a gap between theory and practice.  Indeed, few non-preemptive \emph{online} schedulers are known to have strong foundational guarantees.  This gap is likely due to strong lower bounds on any online algorithm for popular objectives. Indeed, typical worst case analysis approaches, and even resource augmented approaches such as speed augmentation, result in all algorithms having poor performance guarantees. 

This paper considers on-line non-preemptive scheduling problems in the worst-case rejection model where the algorithm is allowed to reject a small fraction of jobs.  By rejecting only a few jobs, this paper shows that the strong lower bounds can be circumvented.  This approach can be used to discover algorithmic scheduling policies with desirable worst-case guarantees.

Specifically, the paper presents algorithms for the following two objectives: minimizing the total flow-time and minimizing the total weighted flow-time plus energy under the speed-scaling mechanism. The algorithms have a small constant competitive ratio while rejecting only a constant fraction of jobs.   

Beyond specific results, the paper asserts that alternative models beyond speed augmentation should be explored to aid in the discovery of good schedulers in the face of the requirement of being online and non-preemptive.  

\end{abstract}

\section{Introduction}
\label{sec:intro}

Designing efficient system schedulers is critical for optimizing system performance.  Many environments require the scheduler to be \emph{non-preemptive}, ensuring each job is scheduled on a machine without interruption.  The need for non-preemption arises because preemption requires saving the state of a program and writing the state to memory or disk.  For large complex tasks, the overhead cost of saving state is so large that it has to be avoided entirely.  

Designing theoretically efficient \emph{online}  non-preemptive schedulers is challenging.  Strong lower bounds have been shown, even for simple instances~\cite{KellererTW99,ChekuriKhanna01:Algorithms-for-minimizing}.  
The difficulty lies in the pessimism of assuming the algorithm is online and must be robust to all problem instances combined with irrevocable nature of scheduling a non-preemptive jobs.


In order to overcome strong theoretical barriers when designing scheduling algorithms, \cite{KalyanasundaramPruhs00:Speed-is-as-powerful} and \cite{PhillipsStein02:Optimal-time-critical} proposed using resource augmentation in terms of \emph{speed augmentation} and the \emph{machine augmentation}, respectively.  The idea is to either give the algorithm faster processors or extra machines versus the adversary.   These models provide a tool to establish a theoretical explanation for the good performance of  algorithms in practice.
Indeed, many practical heuristics have been shown to be competitive in the on-line preemptive model where the algorithm is given resource augmentation.

Non-preemptive environments have resisted the discovery of strong theoretical schedulers.  Specifically, it is known that a non-preemptive algorithm cannot have a small reasonable competitive ratio using only speed or machine augmentation~\cite{LucarelliThang16:Online-Non-preemptive} for the popular average flow time objective. 

Recently, \cite{ChoudhuryDas15:Rejecting-jobs} extended the resource augmentation model to allow \emph{rejection}.  That is, some jobs need not be completed and are rejected. By combining rejection and speed augmentation, \cite{LucarelliThang16:Online-Non-preemptive} gave competitive algorithms for  non-preemptive flow-time problems. An intriguing question is the power of rejection versus resource augmentation.  Is there  a competitive algorithm that only uses rejection? This would establish that theoretically rejection is more powerful since there are lower bounds using resource augmentation.  This paper answers  this question positively. 


\subsection{Models, Problems and Contribution}

\paragraph{Non-Preemptive Total Flow-time Minimization}
In this problem, we are given a set of unrelated machines $\mathcal{M}$ and jobs arrive on-line.
Each job $j \in \mathcal{J}$ is characterized by a \emph{release time} $r_j$ and it takes a different \emph{processing time} $p_{ij}$ if it is executed on each machine $i \in \mathcal{M}$.
The characteristics of each job become known to the algorithm only after its arrival.
The jobs should be scheduled \emph{non-preemptively}, that is a job is considered to be successfully executed only if it is executed on a machine $i \in \mathcal{M}$ for $p_{ij}$ continuous time units.
Given a schedule $\mathcal{S}$, the \emph{completion time} of a job $j \in \mathcal{J}$ is denoted by $C_j$.
Then, its \emph{flow-time} is defined as $F_j=C_j-r_j$, that is the total amount of time during which $j$ remains in the system.
Our goal is to create a non-preemptive schedule that minimizes the total flow-times of all jobs, i.e., $\sum_{j} F_j$.

The problem has been studied in \cite{LucarelliThang16:Online-Non-preemptive} in the model of speed augmentation and rejection.
Specifically, \cite{LucarelliThang16:Online-Non-preemptive} gave a $O(1/(\epsilon_{r} \cdot \epsilon_{s}))$-competitive algorithm that uses machines with speed $(1+\epsilon_{s})$ and reject at most $\epsilon_{r}$-fraction of jobs for arbitrarily small $\epsilon_{r},\epsilon_{s} > 0$.
A natural intriguing question is whether speed augmentation is necessary. 
Our main result answers positively this question.

\begin{theorem} \label{thm:flow}
For the non-preemptive total flow-time minimization problem, 
there exists a $2 \bigl( \frac{1+\epsilon}{\epsilon} \bigr)^{2}$-competitive algorithm that removes at most $2\epsilon$ fraction of the total number of jobs, for any $\epsilon>0$.
\end{theorem}

The design and analysis of the algorithm follow the duality approach.
At the release time of any job $j$, the algorithm defines the dual variables associated to the job and assigns $j$ to some machine based on this definition.
The value of the dual variables associated to $j$ are selected in order to satisfy two key properties:
(i) comprise the marginal increase of the total weighted flow-time due to the arrival of the job --- the property that has been observed \cite{AnandGarg12:Resource-augmentation} and has become more and more popular in dual-fitting for on-line scheduling;
and (ii) capture the information for a future decision of the algorithm whether job $j$ will be completed or rejected.
Moreover, the dual variables are defined so as to stabilize the schedule and allows us to maintain a non-preemptive schedule (even with job arrivals and rejections in the future).

The decision about rejecting a job depends on the load of the recently released jobs that are waiting in the queue of each machine.
The scheduler rejects a job when this load exceeds a given threshold, while the rejected job is not necessarily the one that just arrived and caused the excess in the threshold.
The following lemma, whose proof is given in the Appendix, shows that immediate rejection policies cannot improve the competitive ratio.
\begin{lemma}\label{lem:LB1}
Any $\epsilon$-rejection policy which has to decide the rejection or not of each job immediately upon its arrival, has a competitive ratio of $\Omega(\sqrt{\Delta})$ for the non-preemptive total flow-time minimization problem even on a single machine environment, where $\Delta$ is the ratio of the maximum over the minimum processing time in the instance and $\epsilon > 0$.
\end{lemma}
\begin{proof}
Assume that $1/\epsilon$ jobs of length $L$ are released at time $0$.
Note that the algorithm can reject at most one of them.
Consider the time $t$ where the algorithm schedules the first of these jobs. 
\begin{itemize}[leftmargin=*,topsep=5pt]
\item If $t > L^2$, then the algorithm was waited too long.
Specifically, the solution of the algorithm has a total flow time of at least $(1/\epsilon)L^2 + \sum_{j=1}^{1/\epsilon} j \cdot L = \Theta(L^2)$.
On the other hand, the adversary schedules the jobs sequentially in an arbitrary order starting from time $0$.
Hence, the total flow time in adversary's schedule is equal to $\sum_{j=1}^{1/\epsilon} j \cdot L = \Theta(L)$.
Thus, the competitive ratio in this case is $\Omega(L)$.
\item If $t < L^2$, then starting at time $t$ a job of processing time $1/L$ is released every $1/L$ time until $t+L$.
Thus, there are $\Theta(L^2)$ such small jobs released.
By the definition of the model, the algorithm cannot reject the job which is scheduled at time $t$, and hence the small jobs have to wait until this job is completed at time $t+L$.
Since the algorithm can only reject a constant fraction of the small jobs, it will have a total flow time of $\Omega(L^3)$.
On the other hand, the adversary schedules all small jobs before all big jobs of processing time $L$.
Hence, the total flow time for the small jobs is $\Theta(L^2)$, while for the big jobs the total flow time is $(1/\epsilon)(t+L)+\sum_{j=1}^{1/\epsilon} j \cdot L = \Theta(L^2)$, since $t \leq L^2$.
Thus, the competitive ratio is again $\Omega(L)$.
\end{itemize}
The lemma follows from the fact that $\Delta=L^2$.
\end{proof}

\paragraph{Non-Preemptive Total Flow-time Plus Energy Minimization}
We next consider non-preemptive scheduling in the speed scaling model. 
In this model, each machine $i \in \mathcal{M}$ has a power function of the form $P(s_{i}(t)) = s_{i}(t)^\alpha$, where $s_{i}(t)$ is the speed of the machine $i$ at time $t$ and $\alpha>1$ is a constant parameter (usually $\alpha \in (1,3]$).
Each job $j \in \mathcal{J}$ is now characterized by its \emph{weight} $w_{j}$, its \emph{release date} $r_j$ and, for each machine $i \in \mathcal{M}$, a machine-dependent \emph{volume} of execution $p_{ij}$. 
A \emph{non-preemptive} schedule in the speed-scaling model is a schedule in which each job is processed continuously (without being interrupted) in a machine and a job has a constant speed during its execution.
Note that in the model, it is allowed to process multiple jobs in parallel on the same machine. 
The objective is to schedule jobs non-preemptively so that minimizing the total weighted flow-time plus the energy consumed for all jobs, i.e. $\sum_{j} w_j F_j + \sum_{i} \int_{0}^{\infty} \big(s_i(t)\big)^{\alpha} dt$.

Building upon the resilient ideas and techniques from flow-time minimization, we  derive a competitive algorithm for the problem.
Note that this algorithm does not need to process multiple jobs in parallel on the same machine, although this is permissible by the described model.

\begin{theorem}\label{thm:flow+energy}
For the non-preemptive total weighted flow-time plus energy minimization problem, there exists an O$\left(\left(1+ \frac{1}{\epsilon}\right)^{\frac{\alpha}{\alpha-1}}\right)$-competitive algorithm that rejects jobs of total weight at most an $\epsilon$-fraction of the total weight of all jobs, for any $\epsilon > 0$.
\end{theorem}

\paragraph{Non-Preemptive Energy Minimization}
Subsequently, we consider the non-preemptive energy minimization scheduling problem in the speed scaling model.
The setting is similar to the previous problem but a job $j \in \mathcal{J}$ now has a release date $r_{j}$, a deadline $d_{j}$ and a processing volume $p_{ij}$ if it is assigned to machine $i \in \mathcal{M}$.
Every job has to be processed non-preemptively and to be completed before its deadline. 
The goal is to minimize the total energy consumption $\sum_{i} \sum_{t} P_{i}\big(s_i(t)\big)$ where $P_{i}$ is the power function of machine $i$. (In this case we consider the discrete time setting.)

No competitive algorithm is known in the non-preemptive multiple-machine environment.
Despite of some similarities to the problem of minimizing energy plus flow-time, the main difference is that in the latter, one can make a trade-off between energy and flow-time and derive a competitive algorithm whereas for the energy minimization problem, one has to deal directly with a non-linear objective.   
The critical issue is that no linear program (LP) with relatively small integrality gap was known.
In order to derive a competitive algorithm for this problem, we make use of the primal-dual approach based on configuration LP recently developed in \cite{Thang17:Online-Primal-Dual}.
The approach consists of introducing exponential number of variables to the natural formulation in order to reduce the integrality gap.
Then, in contrast to current rounding techniques based on configuration LPs, the approach maintains greedily a competitive solution in the sense of primal-dual (without solving exponential size LPs).
Interestingly, using this approach, the power functions are not required to be convex (a crucial property for prior analyses) and the competitive ratio is characterized by a notion of smoothness defined as follows.

\begin{definition}
A set function $f: 2^{\mathcal{N}} \rightarrow \mathbb{R}^{+}$ is $(\lambda,\mu)$-\emph{smooth}
if for any set $A = \{a_{1}, \ldots, a_{n}\} \subseteq \mathcal{N}$ and any collection 
$B_{1} \subseteq B_{2} \subseteq \ldots \subseteq B_{n} \subseteq B \subseteq \mathcal{N}$, 
the following inequality holds.
$$
\sum_{i = 1}^{n} \left[ f\bigl( B_{i} \cup a_{i} \bigr) - f\bigl( B_{i} \bigr)\right]
\leq \lambda f\bigl( A \bigr) + \mu f\bigl( B \bigr)
$$
\end{definition}

\begin{theorem}\label{thm:energy}
Assume that all power functions are $(\lambda,\mu)$-smooth.
Then, there is a $\lambda/(1-\mu)$-competitive algorithm. In particular, 
if $P_{i}(s) = s^{\alpha_{i}}$ for $\alpha_{i} \geq 1$ then the algorithm is $\alpha^{\alpha}$-competitive where $\alpha = \max_{i} \alpha_{i}$.
\end{theorem}

In the following lemma, whose proof is given in the Appendix, we consider the case of typical power functions of the form $P(s)=s^{\alpha}$, and we show that the above result is asymptotically optimal as a function of $\alpha$.

\begin{lemma} \label{lem:LB2}
Any deterministic algorithm is at least $(\alpha/9)^{\alpha}$-compe\-ti\-tive for the non-preemptive energy minimization problem even in a single machine environment.
\end{lemma}

\begin{proof}
The construction is inspired by the one in~\cite{LiuLiu16:Optimal-Nonpreemptive}.

Fix a deterministic on-line algorithm \textsc{Alg}.
Without loss of generality, assume that $\alpha$ is an integer. Recall that time interval has (normalized) size at least 1.
The span of job 1 is defined as $r_{1} = 0$ and $d_{1} = 3^{\alpha+1}$.
The adversary \textsc{Adv} specify the span of subsequent jobs depending on the behavior of \textsc{Alg}.
Let $S^{\textsc{Alg}}_{j}$ and $C^{\textsc{Alg}}_{j}$ be the starting time and completion time of job $j$
by algorithm $\textsc{Alg}$. For every $j \geq 1$, once algorithm \textsc{Alg} decides the starting time and the speed of job $j$ (so the completion time), \textsc{Adv} releases immediately job $j+1$ with release date $r_{j+1} = S^{\textsc{Alg}}_{j} + 1$,
deadline $d_{j+1} = C^{\textsc{Alg}}_{j}$, and volume $p_{j+1} = (d_{j+1} - r_{j+1})/3$. The instance ends when
either the number of released jobs equals $\alpha$ or $d_{k} - r_{k} \leq 1$. 

We first observe that by executing every job by speed 1, \textsc{Adv} can process all jobs such that at any moment,
no two jobs are run in parallel (or in other words, there is no overlapping). Specifically, by definition of jobs 
(especially $p_{j+1} = (d_{j+1} - r_{j+1})/3$), \textsc{Adv} can entirely execute with speed 1 a job $j$ 
outside of interval $[S^{\textsc{Alg}}_{j} + 1,C^{\textsc{Alg}}_{j}]$. So there is no overlapping with job $j+1$ and subsequent jobs. 
Hence, as the speed is at most 1, the total energy induced is at most the length of the biggest span, which is $d_{1} - r_{1} = 3^{\alpha + 1}$.

Besides, by the way \textsc{Adv} releases jobs, a job overlaps with all other jobs in the schedule of \textsc{Alg}.
Imagine now each job $j$ is initially represented by a rectangle of size $(d_{j} - r_{j})$ by $1/3$. An algorithm consists in
reshaping it to another rectangle (contracting the width and augmenting the height) 
and place them in appropriate way. Now suppose that there is a job with span $[r_{k},d_{k}]$ satisfy $d_{k} - r_{k} \leq 1$. 
In this case, the total height of all rectangles is at least $\alpha/3$. Otherwise, suppose that the instance releases $\alpha$ jobs.
Then the total height of all rectangles is also at least $\alpha/3$. In both case, the total energy during the span of the last job
is at least $(\alpha/3)^{\alpha} \cdot 1$. 

Hence, the competitive ratio is at least $(\alpha/9)^{\alpha}$. 
\end{proof}

\subsection{Related Work}
For the on-line non-preemptive scheduling problem of minimizing total weighted flow-time, 
any algorithm has at least $\Omega(n)$ competitive ratio, even for single machine where 
$n$ is the number of jobs (as mentioned in \cite{ChekuriKhanna01:Algorithms-for-minimizing}).
In identical machine environments,
\cite{PhillipsStein02:Optimal-time-critical} gave a constant competitive algorithm that uses $m \log P$ 
machines (recall that the adversary uses $m$ machines),
where $P$ is the ratio of the largest to the smallest processing time.
Moreover, an $O(\log n)$-machine $O(1)$-speed algorithm that returns 
the optimal schedule has been presented in \cite{PhillipsStein02:Optimal-time-critical} for the unweighted flow-time objective.
\cite{EpsteinVanStee06} proposed an $\ell$-machines $O(\min\{\sqrt[\ell]{P},\sqrt[\ell]{n}\})$-competitive algorithm for the unweighted case on a single machine. This algorithm is optimal up to a constant factor for constant $\ell$.
Recently, \cite{LucarelliThang16:Online-Non-preemptive} consider the problem
in the model of speed augmentation and rejection. They showed that without rejection, 
no algorithm is competitive even on single machine with speed arbitrarily faster than that of adversary.
Moreover, they gave a scalable $O(1/(\epsilon_{r} \cdot \epsilon_{s}))$-competitive
algorithm that uses machines with speed $(1+\epsilon_{s})$ and reject at most $\epsilon_{r}$ fraction of jobs for arbitrarily small 
$\epsilon_{r},\epsilon_{s} > 0$.

For the on-line non-preemptive scheduling problem of minimizing total weighted flow-time plus energy, 
to the best of our knowledge, no competitive algorithm is known. However, the problem
in the preemptive setting has been widely studied. 
\cite{BansalChan09:Speed-scaling}
gave an $O(\alpha/\log \alpha)$-competitive algorithm for weighted flow-time plus energy 
in a single machine where the energy function is $s^{\alpha}$.
Based on linear programming and dual-fitting, 
\cite{AnandGarg12:Resource-augmentation} proved an $O(\alpha^{2})$-competitive algorithm 
for unrelated machines. Subsequently, Nguyen~\cite{Thang13:Lagrangian-Duality} and
\cite{DevanurHuang14:Primal-Dual} presented an $O(\alpha/\log \alpha)$-competitive algorithms 
for unrelated machines by dual fitting and primal dual approaches, respectively. 

For the on-line non-preemptive scheduling problem of minimizing total energy consumption, 
no competitive algorithm is known. Even in the preemptive scheduling in which migration of jobs
between machines are not allowed, no algorithm with provable performance is given. 
The difficulty, as mentioned earlier, is due to the integrality gap barrier of all currently known formulations.
In single machine where the issue of non-migration does not exist, \cite{BansalKimbrel07:Speed-scaling}  
gave a $2\bigl( \frac{\alpha}{\alpha - 1} \bigr)^{\alpha} e^{\alpha}$-competitive algorithm.
Moreover, \cite{BansalChan12:Improved-Bounds} showed that no deterministic algorithm has 
competitive ratio less than $e^{\alpha-1}/\alpha$. 
\cite{AlbersBampis16:Scheduling-on-power-heterogeneous} considered the case where jobs are allowed to be executed preemptively and migration between machines is permitted.
For this problem, they proposed an algorithm based 
on the \textsc{Average Rate} algorithm \cite{YaoDemers95:A-scheduling-model} 
and they showed a competitive ratio of $(1+\epsilon)(\alpha^{\alpha} 2^{\alpha - 1} + 1)$.  
\section{Minimize Total Flow-time}
\label{sec:flow}


\paragraph{Linear Programming Formulation}
In order to formulate our problem as a linear program, for each job $j \in \mathcal{J}$, machine $i \in \mathcal{M}$ and time $t \geq r_j$, we introduce a binary variable $x_{ij}(t)$ which is equal to one if $j$ is processed on $i$ at time $t$, and zero otherwise.
We use two lower bounds on the flow-time of each job $j \in \mathcal{J}$, assuming that it is dispatched to machine $i$:
its \emph{fractional flow-time} which is defined as $\int_{r_j}^\infty \frac{t - r_j}{p_{ij}} x_{ij}(t) dt$ (see for example \cite{AnandGarg12:Resource-augmentation}), and its processing time $p_{ij}=\int_{r_j}^\infty x_{ij}(t) dt$.
Then, the linear programming formulation for the problem of minimizing the total flow-time follows.

\begin{align*}
\min \sum_{i \in \mathcal{M}} \sum_{j \in \mathcal{J}} & \int_{r_j}^\infty \Big(\frac{t - r_j}{p_{ij}}+ 1\Big) x_{ij}(t) dt \nonumber \\
\sum_{i \in \mathcal{M}} \int_0^\infty \frac{x_{ij}(t)}{p_{ij}} dt &\geq 1 & &\forall j \label{runs}\\
\sum_{j \in \mathcal{J}} x_{ij}(t) &\leq 1 & &\forall i, t  \\
 x_{ij}(t) \in & \{0,1\} & &\forall i, j, t
\end{align*}

Note that the objective value of the above linear program is at most twice that of the optimal non-preemptive schedule.
We relax the above integer linear program by replacing the integrality constraints for each $x_{ij}(t)$ with $0 \leq x_{ij}(t) \leq 1$.
The dual of the relaxed linear program is the following.

\begin{align*}
\max \sum_{j \in \mathcal{J}} \lambda_j &- \sum_{i \in \mathcal{M}} \int_0^\infty \beta_i(t) && dt \nonumber \\
\frac{\lambda_j}{p_{ij}} - \beta_i(t) &\leq \frac{t - r_j}{p_{ij}} + 1 &&\forall i, j, t \\
\lambda_j &\geq 0 &&\forall j \\
\beta_i(t) &\geq 0 &&\forall i, t
\end{align*}

In the \emph{rejection model} considered in this article, we assume that the algorithm is allowed to reject some jobs.
This can be interpreted in the primal linear program by considering only the variables corresponding to the non-rejected jobs, that is the algorithm does not have to satisfy the first constraint for the rejected jobs.

\paragraph{The Algorithm and Definition of Dual Variables}
We next define the scheduling, the rejection and the dispatching policies of our algorithm which is denoted by $\mathcal{A}$.
Let $\epsilon$, $0 < \epsilon < 1$, be an arbitrarily small constant which indicates the fraction of the total number of jobs that will be rejected.
Each job is immediately dispatched to a machine upon its arrival.
Let $U_i(t)$ be the set of \emph{pending} jobs at time $t$ dispatched to machine $i \in \mathcal{M}$, that is the jobs dispatched to $i$ that have been released but not yet completed or rejected at time~$t$.
Moreover, let $q_{ij}(t)$ be the remaining processing time at time $t$ of a job $j \in \mathcal{J}$ which has been dispatched to the machine $i$.

Let $k$ be the job that is executed on machine $i$ at time $t$.
We always consider the jobs in $U_i(t) \setminus \{k\}$ sorted in non-decreasing order with respect to their processing times; in case of ties, we consider the jobs in earliest release time order.
We say that a job $j \in U_i(t) \setminus \{k\}$ precedes (resp. succeeds) a job $\ell \in U_i(t) \setminus \{k\}$ if $j$ appears before (resp. after) $\ell$ in the above order, and we write $j \prec \ell$ (resp. $j \succ \ell$).
We use the symbols $\preceq$ and $\succeq$ to express the fact that $j$ may coincide with $\ell$.
The \emph{scheduling policy} of the algorithm $\mathcal{A}$ is the following: whenever a machine $i \in \mathcal{M}$ becomes idle at a time $t$, schedule on $i$ the job $j \in U_i(t)$ that precedes any other job in $U_i(t)$.

We use two different rules for defining our \emph{rejection policy}.
The first rule handles the arrival of a big group of jobs during the execution of a long job as in \cite{LucarelliThang16:Online-Non-preemptive}.
The second rule simulates and replaces the utility of speed-augmentation.
\begin{description}
\item[Rejection Rule 1]
At the beginning of the execution of a job $j \in \mathcal{J}$ on machine $i$, we introduce a counter $v_j$ which is initialized to zero.
Whenever a job $\ell$ is dispatched to machine $i$ during the execution of $j$, we increase $v_j$ by $1$.
Then, we interrupt and reject the job $j$ the first time when $v_j=\frac{1}{\epsilon}$.
\item[Rejection Rule 2]
For each machine $i \in \mathcal{M}$, we maintain a counter $c_i$ which is initialized to zero at $t=0$.
Whenever a job $j$ is dispatched to a machine $i$, we increase $c_i$ by $1$.
Then, we reject the job with the largest processing time in $U_i(t)\setminus\{k\}$ the first time when $c_i = 1+\frac{1}{\epsilon}$, and we reset $c_i$ to zero.
\end{description}
Let $\mathcal{R}$ be the set of all rejected jobs.
By slightly abusing the notation, we denote the rejection time of a job $j \in \mathcal{R}$ by $C_j$.
Moreover, we define the flow-time of a rejected job $j \in \mathcal{R}$ to be the difference between its rejection time and its arrival time, and we denote it by $F_j$.

At the arrival of a new job $j \in \mathcal{J}$, let $\Delta_{ij}$ be the increase in the total flow-time if we decide to dispatch the job $j$ to the machine $i$.
Fix a machine $i$ and let $k$ be the job that is executed on $i$ at $r_j$.
Then, assuming that $j$ is dispatched to $i$ (i.e., assuming that $j \in U_i(r_j)$), we have that
\begin{align*}
\Delta_{ij} = & \ q_{ik}(r_j) \cdot \mathbbm{1}_\text{\{if $k$ is not rejected (due to Rule 1)\}} + \sum_{\ell \preceq j} p_{i\ell} \\
 & + \sum_{\ell \succ j} p_{ij} \\
 & - \biggl(q_{ik}(r_j) + \sum_{\ell\not=j} q_{ik}(r_j) \biggr) \cdot \mathbbm{1}_{\{\text{if } k \text{ is rejected due to Rule 1}\}}\\
 & - \biggl(q_{ik}(r_j) + \sum_{\ell \not= j} p_{i\ell} + p_{ij_{\max}} \biggr) \cdot \mathbbm{1}_{\{\text{if } j_{\max} \text{ is rejected due to Rule 2}\}}
\end{align*}
where the first term corresponds to the flow-time of the new job $j$,
the second term corresponds to the increase of the flow-time for the jobs in $U_i(r_j)$ due to the dispatching of $j$ to machine $i$,
the third term corresponds to the decrease of the flow-time for the jobs in  $U_i(r_j)\cup\{k\}$ due to the rejection of $k$ (according to Rule 1),
and the forth term corresponds to the decrease of the flow-time of the largest job $j_{\max}$ due to its rejection (according to Rule 2).
Based on the above, we define
\begin{equation*}
\lambda_{ij} = \frac{1}{\epsilon} p_{ij} + \sum_{\ell \preceq j} p_{i\ell} + \sum_{\ell \succ j} p_{ij}
\end{equation*}
Then, our \emph{dispatching policy} is the following: at the arrival of a new job $j$ at time $r_j$, dispatch $j$ to the machine $i^* = \text{argmin}_{i \in \mathcal{M}} \lambda_{ij}$.

The quantity $\lambda_{ij}$ is strongly related with the marginal increase $\Delta_{ij}$.
However, all negative terms that appear in $\Delta_{ij}$ have been eliminated in $\lambda_{ij}$.
Moreover, the positive quantity $q_{ik}(r_j)$ does not appear in $\lambda_{ij}$, but we have added the term $\frac{1}{\epsilon} p_{ij}$.
The intuition for the definition of $\lambda_{ij}$ is to charge an upper bound to the marginal increase $\Delta_{ij}$ to the $\lambda_{i\ell}$ quantities of some jobs dispatched to $i$.
Specifically, the quantity $\sum_{\ell \preceq j} p_{i\ell} + \sum_{\ell \succ j} p_{ij}$ is charged to $\lambda_{ij}$.
If the positive quantity $q_{ik}(r_j)$ exists, then it is charged to the term $\frac{1}{\epsilon} p_{ik}$ of $\lambda_{ik}$ (i.e., to the job $k$ that is executed on $i$ at the arrival of $j$).
The rejection Rule~1 guarantees that this term is sufficient for all jobs arrived and dispatched to $i$ during the execution of $k$.

In order to deal with the ignored negative terms, we expand the notion of completion time of each job $j \in \mathcal{J}$.
Let $D_j$ be the set of jobs that are rejected due to Rule~1 after the release time of $j$ and before its completion or rejection (including $j$ in case it is rejected), that is the jobs that cause a decrease to the flow time of $j$ due to Rule~1.
Moreover, we denote by $j_k$ the job released at the moment we reject a job $k \in \mathcal{R}$.
Then, we say that a job $j \in \mathcal{J}$ which is dispatched to machine $i$ is \emph{definitively finished} at the time
\begin{eqnarray*}
\tilde{C}_j & = & C_j + \sum_{k \in D_j} q_{ik}(r_{j_k}) \\
&& + \biggl(q_{ik}(r_{j_j}) + \sum_{\ell \not= j_j} p_{i\ell} + p_{ij} \biggr) \cdot \mathbbm{1}_{\{\text{if } j \text{ is rejected due to Rule 2}\}}
\end{eqnarray*}
Let $V_i(t)$ be the set of jobs that are completed or rejected at time $t$ but not yet definitively finished.
Intuitively, at the completion or rejection of job $j$ at time $C_j$ is moved from the set of pending jobs $U_i(t)$ to the set of not yet definitively finished jobs $V_i(t)$, and it remains to this set until the time $\tilde{C}_j$.
Let $R_i(t) \subseteq V_i(t)$ be the set of jobs that are already rejected due to Rule~2 at time $t$ but they are not yet definitively finished.

It remains to formally define the dual variables.
At the arrival of a job $j \in \mathcal{J}$, we set $\lambda_j = \frac{\epsilon}{1+\epsilon} \min_{i \in \mathcal{M}} \lambda_{ij}$ and we never change this value again.
Moreover, for each $i \in \mathcal{M}$ and $t \geq 0$, we set $\beta_i(t) = \frac{\epsilon}{(1+\epsilon)^2}(|U_i(t)|+|V_i(t)|)$.
Note that, given any fixed time $t$, $\beta_i(t)$ may increase if a new job arrives at any time $t' < t$.
However, $\beta_i(t)$ never decreases in the case of rejection since the rejected jobs are transferred to the set $V_i(t)$ where they remain until they are definitively finished.

\paragraph*{Analysis}
We first show the following lemma which relates all but $c_i$ jobs in $U_i(t)$ to some jobs in $R_i(t)$.

\begin{lemma} \label{lem:mapping}
Fix a machine $i$ and a time $t$.
Consider the jobs in $R_i(t)$ sorted in non-decreasing order of the time they are definitively finished; let $k_1,k_2,\ldots,k_r$ be this order, where $r=|R_i(t)|$.
There is a partition of the jobs in $U_i(t)$ into at most $r+1$ subsets, $U_i^1(t), U_i^2(t), \ldots, U_i^{r+1}(t)$ such that
\begin{enumerate}[label=(\roman*)]
\item $|U_i^\ell(t)| \leq \frac{1}{\epsilon}$, for $1 \leq \ell \leq r$,
\item $|U_i^{r+1}(t)| \leq c_i$,
\item for each job $j \in U_i^\ell(t)$, $1 \leq \ell \leq r$, the estimated completion time of $j$ assuming that no other job is released after time $t$ is at most $\tilde{C}_{k_\ell}$.
\end{enumerate}
\end{lemma}
\begin{proof}
The proof is based on induction on time.
We consider only times which correspond to discrete events that modify the sets $U_i(t)$ and $R_i(t)$, i.e., arrival of a new job, completion of a job, rejection of a job according to Rule~2 and definitive finish of a job in $R_i(t)$.

At the arrival of the first job dispatched to machine $i$, we have that $c_i=1$ and the statement directly holds.
Let us assume that the partition exists at an event which occurs at time $t$.
We will show that this holds also for the next event at time $t' \geq t$.
We consider the following three cases.
\begin{itemize}[leftmargin=*,topsep=5pt]
\item If a job $j$ completes at time $t'$, then $j$ is removed from $U_i(t')$ without affecting the mapping implied by the statement of the lemma.
\item If a job $j$ arrives at time $t'$, then $c_i$ is increased by one.
Let $j_\ell$, $1 \leq \ell \leq r+1$, be the job with the largest processing time in $U_i^\ell(t)$.
If $p_j \geq p_{j_r}$, then we set $U_i^\ell(t')=U_i^\ell(t)$ for $1 \leq \ell \leq r$ and $U_i^{r+1}(t')=U_i^{r+1}(t) \cup \{j\}$ and the partition is valid since $c_i$ is increased.
Otherwise, find the biggest $z$, $1 \leq z \leq \ell$, such that $p_j < p_{j_z}$.
We set $U_i^\ell(t')=U_i^\ell(t)$ for $1 \leq \ell \leq z-1$, $U_i^z(t')=U_i^z(t) \cup \{j\} \setminus \{j_z\}$, and $U_i^\ell(t')=U_i^\ell(t) \cup \{j_{\ell-1}\} \setminus \{j_\ell\}$ for $z+1 \leq \ell \leq r+1$.
By these definitions, the first two items of the lemma are satisfied by the induction hypothesis since each set, except for $U_i^{r+1}$, has the same size at times $t$ and $t'$.
For item~(iii), we observe that the job that is added in each set $U_i^\ell$, $z \leq \ell \leq r$, has a shorter processing time than the job which is removed.
Hence, the item~(iii) holds by the definition of the scheduling policy.
Moreover, if a job $k$ is rejected according to Rule~2 at time $t'$, then $R_i(t')=R_i(t) \cup \{k\}$ and $U_i^{|R_i(t')|}(t')=U_i^{r+1}(t)$.
Therefore, the lemma holds since $c_i \leq 1+\frac{1}{\epsilon}$ and $k$ is the job with the largest processing time (and hence the largest estimated completion time) in $U_i^{r+1}(t)$.
\item If the job $k_1$ is definitively finished at time $t'$, then assume that $U_i^1(t)$ is not empty.
Then, by the induction hypothesis each job $j \in U_i^1(t)$ should complete before $t'$, which is a contradiction to the fact that $t'$ is the next event after $t$.
\end{itemize}
Therefore, the lemma follows.
\end{proof}

The following corollary is an immediate consequence of Lemma~\ref{lem:mapping}.

\begin{corollary} \label{cor:mapping}
For each $t$, it holds that $|U_i(t)| \leq \frac{1}{\epsilon}(|R_i(t)|+1)$.
\end{corollary}

The following lemma guarantees that the definition of the dual variables lead always to a feasible solution for the dual program.

\begin{lemma}\label{lem:dual}
For all $i \in \mathcal{M}$, $j \in \mathcal{J}$ and $t \geq r_j$,
the dual constraint is feasible.
\end{lemma}
\begin{proof}
For a machine $i$ and a job $j$, observe that for any fixed $t \geq r_j$, the value of $\beta_i(t)$ may only increase during the execution of the algorithm.
Hence, it is sufficient to prove the constraint assuming that no job arrives after $r_j$.
Assume that the job $k$ is executed on the machine $i$ at the arrival of the job $j$.
We have the following cases.
\medskip

\noindent\textbf{Case 1:} The job $k$ is executed at $t$.
By the definition of $\lambda_j$ and $\lambda_{ij}$, we have:
\begin{align*}
\frac{\lambda_j}{p_{ij}} & \leq \frac{\epsilon}{1+\epsilon} \left(\frac{1}{\epsilon} + \frac{1}{p_{ij}} \sum_{\ell \preceq j} p_{i\ell} + \sum_{\ell \succ j} 1\right)
  \leq \frac{\epsilon}{1+\epsilon} \left(\frac{1}{\epsilon} + \sum_{\ell \preceq j} 1 + \sum_{\ell \succ j} 1\right)\\
 & \text{(since $p_{i\ell} \leq p_{ij}$ for all $\ell \preceq j$)} \\
 & \leq \frac{\epsilon}{1+\epsilon} \left(\frac{1}{\epsilon} + |U_i(t)| + \frac{t-r_j}{p_{ij}}\right)\\
 & \text{(since $t-r_j\geq0$)}
\end{align*}

\noindent\textbf{Case 2:} A job $z \preceq j $ is executed at $t$.
Then, we have $t-r_j \geq \sum_{\ell \prec z} p_{i\ell}$.
Using the definition of $\lambda_j$ and $\lambda_{ij}$, we have:
\begin{align*}
\frac{\lambda_j}{p_{ij}} & \leq \frac{\epsilon}{1+\epsilon} \left(\frac{1}{\epsilon} + \frac{1}{p_{ij}} \sum_{\ell \preceq j} p_{i\ell} + \sum_{\ell \succ j} 1\right)\\
 & = \frac{\epsilon}{1+\epsilon} \left(\frac{1}{\epsilon} + \frac{1}{p_{ij}} \sum_{\ell \prec z} p_{i\ell} + \frac{1}{p_{ij}} \sum_{z \preceq \ell \preceq j} p_{i\ell} + \sum_{\ell \succ j} 1\right)\\
 & \leq \frac{\epsilon}{1+\epsilon} \left(\frac{1}{\epsilon} + \frac{t-r_j}{p_{ij}} + \sum_{z \preceq \ell \preceq j} 1 + \sum_{\ell \succ j} 1\right)\\
 & \text{(since $p_{i\ell} \leq p_{ij}$ for all $\ell \preceq j$)} \\
 & \leq \frac{\epsilon}{1+\epsilon} \left(\frac{1}{\epsilon} + \frac{t-r_j}{p_{ij}} + |U_i(t)|\right)
\end{align*}

\noindent\textbf{Case 3:} A job $z \succ j $ is executed at $t$.
Then, we have $t-r_j \geq \sum_{\ell \prec z} p_{i\ell}$.
Using the definition of $\lambda_j$ and $\lambda_{ij}$, we have:
\begin{align*}
\frac{\lambda_j}{p_{ij}} & \leq \frac{\epsilon}{1+\epsilon} \left(\frac{1}{\epsilon} + \frac{1}{p_{ij}} \sum_{\ell \preceq j} p_{i\ell} + \sum_{\ell \succ j} 1\right)\\
 & = \frac{\epsilon}{1+\epsilon} \left(\frac{1}{\epsilon} + \frac{1}{p_{ij}} \sum_{\ell \prec j} p_{i\ell} + \sum_{j \prec \ell \prec z} \frac{p_{i\ell}}{p_{i\ell}} + \sum_{\ell \succeq z} 1\right)\\
 & \leq \frac{\epsilon}{1+\epsilon} \left(\frac{1}{\epsilon} + \frac{1}{p_{ij}} \sum_{\ell \prec j} p_{i\ell} + \sum_{j \prec \ell \prec z} \frac{p_{i\ell}}{p_{ij}} + \sum_{\ell \succeq z} 1\right)\\
 & \text{(since $p_{i\ell} > p_{ij}$ for all $\ell \succ j$)} \\
 & \leq \frac{\epsilon}{1+\epsilon} \left(\frac{1}{\epsilon} + \frac{t-r_j}{p_{ij}} + |U_i(t)|\right)
\end{align*}
\medskip

Hence, in all the three cases we have:
\begin{align*}
\frac{\lambda_j}{p_{ij}} & \leq \frac{\epsilon}{1+\epsilon} \left(\frac{1}{\epsilon} + \frac{t-r_j}{p_{ij}} + |U_i(t)|\right) \\
 & = \frac{\epsilon}{1+\epsilon} \left(\frac{1}{\epsilon} + \frac{t-r_j}{p_{ij}} + \frac{|U_i(t)|+\epsilon|U_i(t)|}{1+\epsilon}\right) \\
 & \leq \frac{\epsilon}{1+\epsilon} \left(\frac{1}{\epsilon} + \frac{t-r_j}{p_{ij}} + \frac{|U_i(t)| + |R_i(t)|+1}{1+\epsilon}\right) \\
 & \text{(by Corollary~\ref{cor:mapping})} \\
 & \leq \frac{1}{1+\epsilon} + \frac{\epsilon}{(1+\epsilon)^2} + \frac{\epsilon}{1+\epsilon} \frac{t-r_j}{p_{ij}} + \beta_i(t) 
  < 1 + \frac{t-r_j}{p_{ij}} + \beta_i(t)
\end{align*}
and the lemma follows.
\end{proof}



Using the above results, we next prove Theorem~\ref{thm:flow}.

\begin{proof}[Proof of Theorem~\ref{thm:flow}]
An immediate consequence of the definition of the two rejection rules is that the jobs rejected by algorithm $\mathcal{A}$ is at most a $2\epsilon$-fraction of the total number of jobs in $\mathcal{J}$.
By Lemma~\ref{lem:dual}, we know that the proposed definition of the dual variables leads to a feasible dual solution.
For the objective value of the dual program, by the definition of $\lambda_j$ and $\tilde{C}_j$, we have that
\begin{equation*}
\sum_{j \in \mathcal{J}} \lambda_j \geq \frac{\epsilon}{1+\epsilon} \sum_{j \in \mathcal{J}} (\tilde{C}_j - r_j)
\end{equation*}
Moreover, by the definition of $U_i(t)$, $V_i(t)$ and $\tilde{C}_j$, we have that
\begin{equation*}
\sum_{i \in \mathcal{M}} \int_0^{\infty} \beta_i(t) = \frac{\epsilon}{(1+\epsilon)^2} \sum_{j \in \mathcal{J}} (\tilde{C}_j - r_j)
\end{equation*}
Then, the dual objective is at least
\begin{equation*}
\left(\frac{\epsilon}{1+\epsilon}\right)^2 \sum_{j \in \mathcal{J}} (\tilde{C}_j - r_j)
\end{equation*}

Let $F_j^{\mathcal{A}}$ be the flow time of a job $j \in \mathcal{J}$ in the schedule constructed by algorithm $\mathcal{A}$;
recall that, for a rejected job $j \in \mathcal{R}$, $F_j^{\mathcal{A}}$ corresponds to the time between its release and its rejection.
By definition, we have that $\tilde{C}_j - r_j \geq F_j^{\mathcal{A}}$, for each $j \in \mathcal{J}$.
Therefore, taking into account that the objective value of our primal linear program is at most twice the value of an optimal
non-preemptive schedule, the theorem follows.
\end{proof}

\section{Minimize Total Weighted Flow Time plus Energy}
\label{sec:flow+energy}

\paragraph{Linear Programming Formulation}
Let $\delta_{ij} = \frac{w_j}{p_{ij}}$ be the \emph{density} of a job $j \in \mathcal{J}$ on machine $i \in \mathcal{M}$. 
Let $s_{ij}(t)$ be a variable that represents the speed at which the job $j \in \mathcal{J}$ is executed on machine $i \in \mathcal{M}$ at time $t$.
Given a constant $\gamma$ that will be defined later, we consider the following convex programming formulation for the problem of minimizing the total weighted flow time plus energy.
\begin{align*}
\min & \sum_{i \in \mathcal{M}} \sum_{j \in \mathcal{J}} \int_{r_j}^{\infty} s_{ij}(t)\delta_{ij} (t-r_j+p_{ij}) dt \nonumber \\
 & + \frac{\alpha}{\gamma(\alpha-1)} \sum_{i \in \mathcal{M}} \sum_{j \in \mathcal{J}} w_{j}^{\frac{\alpha-1}{\alpha}} \int_{r_j}^{\infty} s_{ij}(t) dt \nonumber \\
 & + \sum_{i \in \mathcal{M}} \int_{r_j}^{\infty} \Big(\sum_{j \in \mathcal{J}} s_{ij}(t)\Big)^{\alpha} dt \nonumber \\
\sum_{i \in \mathcal{M}} & \int_{r_j}^{\infty} \frac{s_{ij}(t)}{p_{ij}} dt \geq 1  & \forall j \in \mathcal{J} \\
s_{ij}(&t)  \geq 0 & \forall i \in \mathcal{M}, j\in \mathcal{J}, t\geq r_j
\end{align*}

The first and the second~\cite{AnandGarg12:Resource-augmentation} terms of the objective correspond to the weighted fractional flow time whereas the third term corresponds to the total energy consumed.
In order to linearize the convex energy term, we use the following property which holds for any convex function $f(x)$: $f(x) \geq f(y) + f'(y) (x-y)$.
Thus, we can relax the objective function by replacing its last term by
\begin{equation*}
\sum_{i \in \mathcal{M}} \int_{0}^{\infty} (1 - \alpha)\big(u_{i}(t)\big)^{\alpha} dt 
+ \sum_{i \in \mathcal{M}} \int_{0}^{\infty} \alpha \big(u_{i}(t)\big)^{\alpha-1} \Big(\sum_{j \in \mathcal{J}} s_{ij}(t)\Big) dt
\end{equation*}

Note that the only variables in the above formulation are $s_{ij}(t)$.
The quantities $u_{i}(t)$ are constants that will be defined later.
In fact, $u_{i}(t)$'s will be treated as dual variables and they will be defined during the primal-dual procedure.
The dual of the above LP is the following:
\begin{eqnarray*}
& \max \sum_{j \in \mathcal{J}} \lambda_j  + \sum_{i \in \mathcal{M}} \int_{0}^{\infty} (1-\alpha) \big(u_{i}(t)\big)^{\alpha} dt \nonumber \\
& \frac{\lambda_j}{p_{ij}} \leq \delta_{ij} (t-r_j+p_{ij}) + \alpha\big(u_{i}(t)\big)^{\alpha-1}
 + \frac{\alpha}{\gamma(\alpha-1)} w_{j}^{\frac{\alpha-1}{\alpha}}\\
& \hfill \forall i \in \mathcal{M}, j \in \mathcal{J}, t \geq r_j
\end{eqnarray*}

\paragraph{The Algorithm and Definition of Dual Variables}
In this section, we define the scheduling, the rejection and the dispatching policies of our algorithm which is denoted by $\mathcal{A}$. 
Let $0 < \epsilon < 1$ be some arbitrarily small constant which corresponds to the fraction of the rejected weights.
Each job is immediately dispatched to some machine $i \in \mathcal{M}$ upon its arrival.
Let $U_i(t)$ be the set of \emph{pending} jobs at time $t$ dispatched to machine $i \in \mathcal{M}$, that is the jobs dispatched to $i$ that have been released but not yet completed or rejected at time $t$.
Moreover, let $q_{ij}(t)$ be the remaining volume at time $t$ of job $j$ which is dispatched to machine $i$.

Let $k$ be the job that is being executed on machine $i$ at time $t$.
We consider the jobs in $U_i(t) \setminus \{k\}$ sorted in non-increasing order with respect to their densities; in case of ties, we consider the jobs in earliest release time order.
We say that a job $j \in U_i(t) \setminus \{k\}$ precedes (resp. succeeds) a job $\ell \in U_i(t) \setminus \{k\}$ if $j$ appears before (resp. after) $\ell$ in the above order, and we write $j \prec \ell$ (resp. $j \succ \ell$).
We use the symbols $\preceq$ and $\succeq$ to express the fact that $j$ may coincide with $\ell$.

The \emph{scheduling policy} of the algorithm $\mathcal{A}$ is the following: whenever a machine $i \in \mathcal{M}$ becomes idle at a time $t$, schedule on $i$ the job $j \in U_i(t)$ that precedes any other job in $U_i(t)$.
The \emph{speed} of the machine $i$ at the start time $j$ is defined as $s_{ij} = \gamma \Big(\sum_{\ell \in U_i(t)} w_{\ell}\Big)^{1/\alpha}$.
Note that, the speed of $i$ is defined at the beginning of the execution of $j$ and does not change until $j$ is completed or rejected.
Assuming that no other jobs arrive in the future, we can compute the expected speed of each remaining pending job $\ell \in U_i(t)$ which is equal to $\gamma\Big(\sum_{\ell' \succeq \ell} w_{\ell'}\Big)^{1/\alpha}$.

As soon as the machine $i$ starts executing a job $j$, we introduce a counter $v_j$ which is initialized to zero.
Each time a job $\ell$ is released during the execution of $j$ and it is dispatched to machine $i$, we increase $v_j$ by $w_\ell$.
Then, the \emph{rejection policy} of the algorithm $\mathcal{A}$ is the following: interrupt the execution of $j$ and reject it the first time when $v_j > w_j/\epsilon$.

Assume that at the arrival of a new job $j$ at time $r_j$, the machine $i$ is executing the job $k$.
For each $\ell \in U_i(t)\setminus\{k\}$, let $W_\ell = \sum_{\ell' \in U_i(t)\setminus\{k\}: \ell' \succeq \ell} w_{\ell'}$.
We denote by $\Delta_{ij}$ the marginal increase in the total weighted flow time that will occur following the scheduling and rejection policies of $\mathcal{A}$, if we decide to dispatch the job $j$ to machine $i$.
Then, $\Delta_{ij}$ can be bounded as follows (we ignore the increase of the speed and hence the decrease of the processing time for each job $\ell \prec j$)
\begin{equation*}
\Delta_{ij} \leq
\begin{cases}
\displaystyle
 w_j \left( \frac{q_{ik}(r_j)}{s_k} + \sum_{\ell \preceq j} \frac{p_{i\ell}}{\gamma W_\ell^{1/\alpha}} \right)
 + \Big(\sum_{\ell \succ j} w_\ell\Big) \frac{p_{ij}}{\gamma W_j^{1/\alpha}}\\
 \hfill \text{if } v_k + w_j \leq \frac{w_k}{\epsilon} \\
\displaystyle
 w_j \left( \sum_{\ell \preceq j} \frac{p_{i\ell}}{ \gamma W_\ell^{1/\alpha}} \right) 
 + \Big(\sum_{\ell \succ j} w_\ell\Big) \frac{p_{ij}}{\gamma W_j^{1/\alpha}}
 - \Big(\sum_{\ell \not= j} w_\ell\Big) \frac{q_{ik}(r_j)}{s_k} \\
 \hfill \text{otherwise}
\end{cases}
\end{equation*}
where in both cases, the first positive term correspond to the weighted flow time of the job $j$, while the second positive term correspond to the marginal increase of the weighted flow time of other jobs, that is the completion time of the jobs with density smaller than the density of $j$ is delayed by $p_{ij}/\gamma W_j^{1/\alpha}$.
The negative term in the second case corresponds to the decrease in the weighted flow time of all jobs in $U_i(t)$ if the job $k$ is rejected.
Then, we define a set of variables $\lambda_{ij}$, for all $i \in \mathcal{M}$, as:
$\lambda_{ij} 
 = w_j \left(\frac{p_{ij}}{\epsilon} + \sum_{\ell \preceq j} \frac{p_{i\ell}}{\gamma W_\ell^{1/\alpha}} \right) 
 + \Big(\sum_{\ell \succ j} w_\ell\Big) \frac{p_{ij}}{\gamma W_j^{1/\alpha}}$.
The \emph{dispatching policy} is the following: dispatch the job $j$ to the machine $i^*$ such that $i^* = \text{argmin}_{i \in \mathcal{M}} \{\lambda_{ij}\}$.
\bigskip

We next define the dual variables $\lambda_j$ as well as the quantities $u_{i}(t)$. 
Based on the dispatching policy, we set $\lambda_j = \frac{\epsilon}{1+\epsilon} \min_{i \in \mathcal{M}} \{\lambda_{ij}\}$.
For each job $j$, let $D_j$ be the set of the jobs rejected due to the rejection policy between $r_j$ and the time when $j$ is completed or rejected.
Let $j_k$ denote the job released at the time when our policy rejects the job $k$.
Then, we say that a job $j$ is \emph{definitively finished} at the time $\sum_{k \in D_j} \frac{q_{ik}(r_{j_k})}{s_k}$ after its completion or rejection.
For every job $\ell$, define the fractional weight $w_\ell(t)$ of $\ell$ at time $t$ as $w_\ell q_{i\ell}(t)/p_{i\ell}$. Let $Q_i(t)$ be the set of jobs that are dispatched to machine $i$ and are already completed or rejected but no yet definitively finished at time $t$.
Let $V_i(t) = \sum_{\ell \in U_i(t) \cup Q_i(t)} w_\ell(t)$ be the total \emph{fractional} weight of jobs that are not definitively finished on machine $i$ at time $t$.
We define $u_{i}(t)$ as follows:
$u_{i}(t) = \left(\frac{\epsilon}{\gamma(1+\epsilon)(\alpha-1)}\right)^{\frac{1}{\alpha-1}} V_i(t)^{1/\alpha}$.
Note that when a job is rejected, it is transferred from $U_i(t)$ to $Q_i(t)$ where it remains until the time it is definitively finished. 

Consider now two sets of jobs $I_1$ and $I_2$ assigned to machine $i$ such that they are identical except that there is only a job $j \in I_1 \setminus I_2$. Moreover, assume that no job is released after time $r_j$ in either of the instances. Then the algorithm $\mathcal{A}$ is said to be \emph{monotonic} iff $\sum\limits_{l\in I_2} w_l(t) \leq \sum\limits_{l \in I_1} w_l(t), \forall t$ where the jobs in $I_1$ and $I_2$ are scheduled according to $\mathcal{A}$.
The following lemma shows the monotonicity of $V_i(t)$.

\begin{lemma} \label{lemma:monotonicity}
$V_i(t)$ is monotone for every machine $i$.
\end{lemma}
\begin{proof}
Let $k$ be the job executing on machine $i$ at time $t$. Observe that $V_i(t)$ changes due to the arrival of a new job. Assume that a new job $j$ arrives at $t$, i.e. $t = r_j$. Then, it is sufficient to show that $V_i(t)$ is non-decreasing during anytime $t' \geq t$. Consider the jobs in $U_i(t) \setminus \{k\}$. Since all such jobs are scheduled in  non-increasing order of their densities, the total fractional weight of jobs in $U_i(t) \setminus \{k\}$ is monotonic with respect to arrival of a new job (refer to Lemma 6.1 in~\cite{AnandGarg12:Resource-augmentation}). 

Now we consider the case if $k$ is rejected or not rejected at time $t$. In the case $k$ is not rejected then for $t' < t + \frac{q_{ik}(r_j)}{s_j}$, the speed of the machine $i$ is a constant. Hence, $U_i(t')$ is a constant. Using Lemma 6.1 in~\cite{AnandGarg12:Resource-augmentation}, the lemma holds for this case. In the case $k$ is rejected then $U_i(t)$ decreases due to the removal of $k$. Since all jobs in $U_i(t) \setminus \{j\}$ remain for at least $\frac{q_{ik}(r_j)}{s_j}$ time in $Q_i(t)$ after their completion or rejection from $U_i(t)$, the total fractional weight of jobs in $U_i(t) \cup Q_i(t)$ is monotonic with respect to the rejection of job $k$. Using this property with Lemma 6.1 in~\cite{AnandGarg12:Resource-augmentation}, the lemma holds. 
\end{proof}

\paragraph{Analysis}

The following lemma guarantees that the definition of the dual variables lead always to a feasible solution for the dual program.

\begin{lemma} \label{lemma:feas-flow+energy}
For every $i \in \mathcal{M}$, $j \in \mathcal{J}$ and $t \geq r_j$, the dual constraint is feasible.
\end{lemma}
\begin{proof}
Fix a machine $i$.
By Lemma~\ref{lemma:monotonicity}, $u_{i}(t)$'s do never decrease during the execution of the algorithm.
Hence, it is sufficient to prove the inequality for the job $j$ at time $r_j$.
Let $k$ be the job executed in machine $i$ at $r_j$ .
Moreover, let $\bar{C}_j$ be the completion time of the job $j$ estimated at time $r_j$ if it is assigned to machine $i$.
Specifically, if $k$ is rejected then $\bar{C}_j = r_j + \sum_{\ell \preceq j} \frac{p_{i\ell}}{\gamma W_\ell^{1/\alpha}}$; otherwise we have $\bar{C}_j = r_j + \frac{q_{ik}(r_j)}{s_k} + \sum_{\ell \preceq j} \frac{p_{i\ell}}{\gamma W_\ell^{1/\alpha}}$.

By the definitions of $\lambda_j$ and $\lambda_{ij}$, we have:
\begin{equation*}
\frac{\lambda_j}{p_{ij}} 
 \leq \frac{\epsilon}{1+\epsilon} \frac{\lambda_{ij}}{p_{ij}} 
 = \frac{\epsilon}{1+\epsilon} \left( \frac{w_j}{p_{ij}} \left(\frac{p_{ij}}{\epsilon} + \sum_{\ell \preceq j} \frac{p_{i\ell}}{\gamma W_\ell^{1/\alpha}} \right) + \Big(\sum_{\ell \succ j} w_\ell\Big) \frac{1}{\gamma W_j^{1/\alpha}} \right)
\end{equation*}

Let $w_n$ denote the the weight of the latest job according to the precedence order defined above. 

\noindent \textbf{Case 1: $t \leq \bar{C}_j$.}
Assume, first, that the job $k$ is running at time $t$.  
Hence, we have that
\begin{equation*}
t - r_j = \frac{q_{ik}(r_j) - q_{ik}(t)}{s_k}
\end{equation*}
and thus
\begin{align*}
\frac{\lambda_j}{p_{ij}}& - \delta_{ij}(t - r_j + p_{ij}) \\
& \leq \frac{\epsilon}{1+\epsilon} \left( \frac{w_j}{p_{ij}} \left(\frac{p_{ij}}{\epsilon} + \sum_{\ell \preceq j} \frac{p_{i\ell}}{\gamma W_\ell^{1/\alpha}} \right) + \Big(\sum_{\ell \succ j} w_\ell\Big) \frac{1}{\gamma W_j^{1/\alpha}} \right) \\
& \quad - \frac{w_{j}}{p_{ij}} \left( \frac{q_{ik}(r_j) - q_{ik}(t)}{s_k} + p_{ij} \right) \\
& \leq \frac{\epsilon}{1+\epsilon} \left( \frac{w_j}{p_{ij}} \sum_{\ell \preceq j} \frac{p_{i\ell}}{\gamma W_\ell^{1/\alpha}} + \Big(\sum_{\ell \succ j} w_\ell\Big) \frac{1}{\gamma W_j^{1/\alpha}} \right) \\
& \quad - \frac{w_{j}}{p_{ij}} \cdot \frac{q_{ik}(r_j) - q_{ik}(t)}{s_k} \\
& \leq \frac{\epsilon}{1+\epsilon} \left( \frac{w_j}{p_{ij}} \sum_{\ell \preceq j} \frac{p_{i\ell}}{\gamma W_\ell^{1/\alpha}} + \Big(\sum_{\ell \succ j} w_\ell\Big) \frac{1}{\gamma W_j^{1/\alpha}} \right) \\
& \text{(since } t \geq r_j \text{ and hence } q_{ik}(r_j) - q_{ik}(t) \geq 0 \text{)}\\
& \leq \frac{\epsilon}{1+\epsilon} \left( \sum_{\ell \preceq j} \frac{w_\ell}{\gamma W_\ell^{1/\alpha}} + \Big(\sum_{\ell \succ j} w_\ell\Big) \frac{1}{\gamma W_j^{1/\alpha}} \right) \\
& \text{(since } \frac{w_\ell}{p_{i\ell}} \geq \frac{w_j}{p_{ij}} \text{ for any } \ell \preceq j {)}\\
& \leq \frac{\epsilon}{1+\epsilon} \left( \sum_{\ell \preceq j} \frac{w_\ell}{\gamma W_\ell^{1/\alpha}} + \sum_{\ell \succ j} \frac{w_\ell}{\gamma W_\ell^{1/\alpha}} \right) \\
& = \frac{\epsilon}{1+\epsilon} \sum_{\ell \not= k} \frac{w_\ell}{\gamma W_\ell^{1/\alpha}} \\
& \leq \frac{\epsilon}{1+\epsilon} \int_{w_n}^{V_i(t)+w_j} \frac{dz}{\gamma z^{1/\alpha}} \\
& = \frac{\epsilon}{1+\epsilon} \cdot \frac{\alpha}{\gamma(\alpha-1)} \big(V_i(t) + w_j\big)^{\frac{\alpha-1}{\alpha}} \\
& \leq \frac{\epsilon}{1+\epsilon} \cdot \frac{\alpha}{\gamma(\alpha-1)} \left(V_i(t)^{\frac{\alpha-1}{\alpha}} + w_j^{\frac{\alpha-1}{\alpha}}\right) \\
& = \frac{\epsilon}{1+\epsilon} \cdot \frac{\alpha}{\gamma(\alpha-1)} \left(\frac{\gamma(1+\epsilon)(\alpha-1)}{\epsilon} \big(u_{i}(t)\big)^{\alpha-1} + w_j^{\frac{\alpha-1}{\alpha}} \right) \\
& = \alpha \big(u_{i}(t)\big)^{\alpha-1} + \frac{\epsilon}{1+\epsilon} \cdot \frac{\alpha}{\gamma(\alpha-1)} w_j^{\frac{\alpha-1}{\alpha}} \\
& \leq \alpha \big(u_{i}(t)\big)^{\alpha-1} + \frac{\alpha}{\gamma(\alpha-1)} w_j^{\frac{\alpha-1}{\alpha}}
\end{align*}

Assume now that a job $h \not= k$ is executing at time $t$.
Therefore, the machine $i$ has processed all the jobs which have density higher than $\delta_{ih}$.
Moreover, the job $k$ is either completed or rejected.
Hence, we have that
\begin{equation*}
t - r_j \geq \sum_{\ell \prec h} \frac{p_{i\ell}}{\gamma{W_\ell}^{1/\alpha}} + \frac{p_{ih} - q_{ih(t)}}{\gamma{W_h}^{1/\alpha}}
\end{equation*}
and thus
\begin{align*}
\frac{\lambda_j}{p_{ij}}& - \delta_{ij} (t - r_j + p_{ij}) \\
& \leq \frac{\epsilon}{1+\epsilon} \left( \frac{w_j}{p_{ij}} \left(\frac{p_{ij}}{\epsilon} + \sum_{\ell \preceq j} \frac{p_{i\ell}}{\gamma W_\ell^{1/\alpha}} \right) + \Big(\sum_{\ell \succ j} w_\ell\Big) \frac{1}{\gamma W_j^{1/\alpha}} \right) \\
& \quad - \frac{w_{j}}{p_{ij}} \left( \sum_{\ell \prec h} \frac{p_{i\ell}}{\gamma{W_\ell}^{1/\alpha}} + \frac{p_{ih} - q_{ih(t)}}{\gamma{W_h}^{1/\alpha}} + p_{ij} \right) \\
& = \frac{\epsilon}{1+ \epsilon} \left( \frac{w_j}{p_{ij}} \sum_{h \preceq \ell \preceq j} \frac{p_{i\ell}}{\gamma W_\ell^{1/\alpha}} + \Big(\sum_{\ell \succ j} w_\ell\Big) \frac{1}{\gamma W_j^{1/\alpha}} \right) \\
& \quad - \frac{w_j}{p_{ij}} \left( \frac{1}{1+\epsilon} \cdot \sum_{\ell \prec h} \frac{p_{i\ell}}{\gamma W_\ell^{1/\alpha}} - \frac{\epsilon}{1+\epsilon} p_{ij} - \frac{p_{ih} - q_{ih(t)}}{\gamma{W_h}^{1/\alpha}} \right) \\
& \leq \frac{\epsilon}{1+\epsilon} \left( \sum_{h \preceq \ell \preceq j} \frac{w_\ell}{\gamma W_\ell^{1/\alpha}} + \sum_{\ell \succ j} \frac{w_\ell}{\gamma W_\ell^{1/\alpha}} \right) \\
& = \frac{\epsilon}{1+\epsilon} \sum_{\ell \succeq h} \frac{w_\ell}{\gamma W_\ell^{1/\alpha}} \quad
\leq \frac{\epsilon}{1+\epsilon} \int_{w_n}^{V_i(t)+w_j} \frac{d_z}{\gamma z^{1/\alpha}} \\
& \leq \alpha \big(u_{i}(t)\big)^{\alpha-1} + \frac{\alpha}{\gamma(\alpha-1)} w_j^{\frac{\alpha-1}{\alpha}}
\end{align*}

\noindent \textbf{Case 2: $t > \bar{C}_j$.} 
Let $h$ be the job executing at time $t$. Thus, the machine $i$ has processed all the jobs which have density higher than $\delta_{ih}$. Hence, we have 

\begin{align*}
t - r_j &\geq \sum\limits_{\ell \prec h} \frac{p_{i \ell}}{\gamma W_\ell^{1/\alpha}} + \frac{p_{ih} - q_{ih}(t)}{\gamma W_h^{1/\alpha}}
\end{align*}

Thus 
\begin{align*}
\frac{\lambda_{j}}{p_{ij}}& - \delta_{ij}(t-r_j + p_{ij}) \\
&\leq \frac{\epsilon}{1+\epsilon} \left( \frac{w_j}{p_{ij}} \left(\frac{p_{ij}}{\epsilon} + \sum_{\ell \preceq j} \frac{p_{i\ell}}{\gamma W_\ell^{1/\alpha}} \right) + \Big(\sum_{\ell \succ j} w_\ell\Big) \frac{1}{\gamma W_j^{1/\alpha}} \right) \\
& \quad - \frac{w_{j}}{p_{ij}} \left( \sum_{\ell \prec h} \frac{p_{i\ell}}{\gamma{W_\ell}^{1/\alpha}} + \frac{p_{ih} - q_{ih(t)}}{\gamma{W_h}^{1/\alpha}} + p_{ij} \right) \\
& = \frac{\epsilon}{1+ \epsilon} \left( \frac{w_j}{p_{ij}} \sum_{\ell \preceq j} \frac{p_{i\ell}}{\gamma W_\ell^{1/\alpha}} + \Big(\sum_{\ell \succ j} w_\ell\Big) \frac{1}{\gamma W_j^{1/\alpha}} \right) \\
& \quad - \frac{w_j}{p_{ij}} \left( \frac{1}{1+\epsilon} \cdot \sum_{\ell \prec h} \frac{p_{i\ell}}{\gamma W_\ell^{1/\alpha}} - \frac{\epsilon}{1+\epsilon} p_{ij} - \frac{p_{ih} - q_{ih(t)}}{\gamma{W_h}^{1/\alpha}} \right) \\
& \leq \frac{\epsilon}{1+\epsilon} \left( \sum_{\ell \succ h} \frac{w_\ell}{\gamma W_\ell^{1/\alpha}} \right) \\
& = \frac{\epsilon}{1+\epsilon} \sum_{\ell \succeq h} \frac{w_\ell}{\gamma W_\ell^{1/\alpha}} \quad
\leq \frac{\epsilon}{1+\epsilon} \int_{w_n}^{V_i(t)+w_j} \frac{d_z}{\gamma z^{1/\alpha}} \\
& \leq \alpha \big(u_{i}(t)\big)^{\alpha-1} + \frac{\alpha}{\gamma(\alpha-1)} w_j^{\frac{\alpha-1}{\alpha}}
\end{align*}
\end{proof}

Based on this lemma we can prove Theorem~\ref{thm:flow+energy}.

\begin{proof}[Proof of Theorem~\ref{thm:flow+energy}]
By Lemma~\ref{lemma:feas-flow+energy}, the proposed dual variables constitute a feasible solution for the dual program.
Since each job $j \in \mathcal{J}$ is charged to at most one other job while a job $k$ is rejected the first time where $v_k > \frac{w_k}{\epsilon}$, the algorithm $\mathcal{A}$ rejects jobs of total weight at most $\epsilon \sum_{j \in \mathcal{J}} w_j$.
Hence, it remains to give a lower bound for the dual objective based on the proposed dual variables.

Let $\mathcal{R}$ be the set of rejected jobs.
We denote by $F_j^{\mathcal{A}}$ the flow-time of a job $j\in \mathcal{J}\setminus\mathcal{R}$ in the schedule of $\mathcal{A}$.
By slightly abusing the notation, for a job $k \in \mathcal{R}$, we will also use $F_k^{\mathcal{A}}$ to denote the total time passed after $r_k$ until deciding to reject a job $k$, that is, if $k$ is rejected at the release of the job $j \in \mathcal{J}$ then $F_k^{\mathcal{A}}=r_j-r_k$.
Denote by $j_{k}$ the job released at the moment we decided to reject $k$, i.e., for the counter $v_{k}$ before the arrival of job $j_{k}$ we have that $w_{k}/\epsilon - w_{j_{k}}< v_{k} < w_{k}/\epsilon$.

Let $\Delta_j$ be the total increase in the flow-time caused by the arrival of the job $j \in \mathcal{J}$,
i.e., $\Delta_j = \Delta_{ij}$, where $i \in \mathcal{M}$ is the machine to which $j$ is dispatched by $\mathcal{A}$.
For the objective function of the dual program we have

\begin{align*}
\sum_{j \in \mathcal{J}} & \lambda_j  + \sum_{i \in \mathcal{M}} \int_{0}^{\infty} (1-\alpha) \big(u_{i}(t)\big)^{\alpha} dt \\
& \geq \frac{\epsilon}{1+\epsilon} \left( \sum_{j \in \mathcal{J}} \Delta_j + \sum_{k \in \mathcal{R}} \left( \frac{q_{ik}(r_{j_k})}{s_k} \sum_{\ell \not= j_k} w_{\ell} \right) \right) \\
& \quad - (\alpha-1) \left(\frac{\epsilon}{\gamma(1+\epsilon)(\alpha-1)} \right)^{\frac{\alpha} {\alpha-1}} V_i(t) \\
&\geq \left( \frac{\epsilon}{1+\epsilon} - (\alpha-1) \left(\frac{\epsilon}{\gamma(1+\epsilon)(\alpha-1)} \right)^{\frac{\alpha} {\alpha-1}} \right) F^*
\end{align*}

The total weighted flow time plus energy is

$$2F^* +  \left(\frac{\alpha}{\gamma(\alpha-1)}\right)F^* + \sum_i \int_0^\infty (s_i(t))^{\alpha} dt
\leq
\left(2 + \left(\frac{\alpha}{\gamma(\alpha-1)}\right) + \gamma^{\alpha}\right) F^*$$.

Hence the competitive ratio is:
\begin{align*}
    & \frac{\left(2 + \left(\frac{\alpha}{\gamma(\alpha-1)}\right) + \gamma^{\alpha}\right)}{\left( \frac{\epsilon}{1+\epsilon} \right) -  \left(\frac{\epsilon}{\gamma(1+\epsilon)}\right)^{\frac{\alpha}{\alpha-1}} (\alpha-1)^{\frac{-1}{\alpha-1}}}
\end{align*}

We choose $\gamma = \left(\frac{\epsilon}{1+\epsilon}\right)^{\frac{1}{\alpha-1}} \frac{1}{\alpha-1}(\alpha-1 + \ln(\alpha-1))^{\frac{\alpha-1}{\alpha}}$. Observe that denominator becomes $\frac{\epsilon}{1+\epsilon}(\frac{\ln(\alpha-1)}{\alpha-1+\ln(\alpha-1)})$ and the numerator becomes $2 + 2 \left(\frac{1+ \epsilon}{\epsilon}\right)^{\frac{1}{\alpha-1}}+\left(\frac{\epsilon}{1+\epsilon}\right)^2$. Hence the competitive ratio is at most 
O$\left(\left(1+ \frac{1}{\epsilon}\right)^{\frac{\alpha}{\alpha-1}}\right)$.
\end{proof}
\section{Minimize Total Energy Consumption}
\label{sec:energy}

\paragraph{Formulation.}
In the problem, we consider the sets of discretized speeds $\mathcal{V}$ and times. 
We can do that and loose only a factor $(1+\epsilon)$ for $\epsilon$ arbitrarily small.
In the non-preemptive model, the execution of a job is specified by three 
parameters: (1) a machine in which it is executed; (2) a starting time; 
and (3) a speed which is constant during its execution. Note that the
parameters imply the completion time of job. A \emph{valid} execution
of a job $j$ must have the starting time and completion time in $[r_j,d_j]$. 
We say that a \emph{strategy} of a job is a specification of a valid 
execution of the job. Formally, a strategy $s_{i,j,k}$ of a job 
$j$ in machine $i$ indicates the starting time of the job and its speed during the execution. 
Let $\mathcal{S}_j$ be a set of strategies of job $j$.
As the sets of speeds and times are finite, so is the set of 
strategies $\mathcal{S}_j$ for every job $j$.
Let $x_{i,j,k}$ be a variable indicating whether job $j$ is executed by 
strategy $s_{i,j,k} \in \mathcal{S}_{j}$.  
We say that $A$ is a \emph{configuration} in machine $i$ if $A$ is a 
feasible schedule of a subset of jobs. Specifically, $A$ consists of 
tuples $(i,j,k)$ meaning that job $j$ is executed in machine $i$ following 
the strategy $s_{i,j,k}$. 
For configuration $A$ and machine $i$, let $z_{i,A}$ be a variable 
such that $z_{i,A} = 1$ if and only if 
for every triple $(i,j,k) \in A$, $x_{i,j,k} = 1$ and. 
In other words, $z_{i,A} = 1$ iff the schedule in machine $i$ is exactly $A$. 
The energy cost of a configuration $A$ of machine $i$ is $f_{i}(A) = \sum_{t} P_{i}(A(t))$
where $A(t)$ is the speed of the corresponding schedule at time $t$. 
We consider the following formulation and the dual of its relaxation.

\begin{align*}
\min  \sum_{i,A} f_{i}(A)& z_{i,A} \\
\sum_{i,k:s_{i,j,k} \in \mathcal{S}_{j}} x_{i,j,k} &= 1 & &  \forall j \\
\sum_{A: (i,j,k) \in A} z_{i,A}  &= x_{i,j,k}	& & \forall i, j, k \\
\sum_{A} z_{i,A} &= 1 & & \forall i \\
x_{i,j,k}, z_{i,A} &\in \{0,1\} & & \forall i,j,k,A \\
\end{align*}
\begin{align*}
\max \sum_{j} \delta_{j} &+ \sum_{i} \gamma_{i} \\
\delta_{j} &\leq \beta_{i,j,k}  & &  \forall i,j,k \\
\gamma_{i} + \sum_{(i,j,k) \in A} \beta_{i,j,k} &\leq f_{i}(A)  & & \forall i,A \\
\end{align*}

In the primal, the first constraint guarantees that a job $j$ has to be processed 
by some valid execution (in some machine). 
The second constraint ensures that if job $j$ follows strategy $s_{i,j,k}$ then
in the solution, the schedule (configuration) on machine $i$ must contain the 
execution corresponding to strategy $s_{i,j,k}$.
The third constraint says that in the solution, there is always a configuration
(schedule) associated to machine $i$. 

\paragraph{Algorithm.} 
We first interpret intuitively the dual variables, dual constraints and derive useful observations for a
competitive algorithm. Variable $\delta_{j}$ represents the increase of energy to the arrival of job $j$.
Variable $\beta_{i,j,k}$ stands for the marginal energy if job $j$ follows strategy $s_{i,j,k}$.
By this interpretation, the first dual constraint clearly indicates the greedy behavior of an algorithm. 
That is, if a new job $j$ is released, select a strategy $s_{i,j,k} \in \mathcal{S}_{j}$ that minimizes the marginal increase
of the total energy. 

Let $A^{*}_{i}$ be the set of current schedule of machine $i$. Initially, $A^{*}_{i} \gets \emptyset$ for every $i$. 
At the arrival of job $j$, select a strategy $s_{i,j,k} \in \mathcal{S}_{j}$ that minimizes
$
\bigl[ f_{i}(A^{*}_{i} \cup s_{i,j,k}) - f_{i}(A^{*}_{i}) \bigr]  
$
where $(A^{*}_{i} \cup s_{i,j,k})$ is the current schedule with additional execution of 
job $j$ which follows strategy $s_{i,j,k}$. Let $s_{i^{*},j,k^{*}}$ be an optimal strategy.
Then assign job $j$ to machine $i^{*}$ and process it according to the corresponding 
execution of $s_{i^{*},j,k^{*}}$. In the algorithm, we never interrupt or modify the speed of 
a job. 

In fact, we can implement this algorithm as follows.
Let $u_{it}$ be the speed of machine $i$ at time $t$.
Initially, set $u_{it} \gets 0$ for every machine $i$ and time $t$. 
At the arrival of a job $j$, compute the minimum energy increase if job $j$ is assigned to machine $i$
and is executed with constant speed during its execution. Specifically,
it is an optimization problem 
\begin{align*}	
\min_{i} \min_{\tau,v} \sum_{t = \tau}^{\tau + p_{ij}/v} \biggl[ f_{i}\bigl(u_{it} + v\bigr) - f_{i}\bigl(u_{it}\bigr) \biggr]\\
\quad 
\text{s.t}
\quad
r_{j} \leq \tau \leq \tau + \frac{p_{ij}}{v} \leq d_{j}, \quad v \in \mathcal{V}
\end{align*}

\paragraph{Dual variables.} 
Assume that all energy power functions $f_{i}$ are $(\lambda,\mu)$-smooth for some fixed parameters $\lambda > 0$ and $\mu < 1$. 
We are now constructing a dual feasible solution. Define $\delta_{j}$ as $1/\lambda$ times the 
the increase of the total cost due to the arrival of job $j$.
For each machine $i$ and job $j$, define 
$
\beta_{i,j,k} := \frac{1}{\lambda} \bigl[ f_{i}(A^{*}_{i,\prec j} \cup s_{i,j,k}) - f_{i}(A^{*}_{i,\prec j})  \bigr]
$
where $A^{*}_{i, \prec j}$ is the schedule of machine $i$ (due to the algorithm) prior to the arrival of job $j$.
Finally, for every machine $i$ define dual variable 
$
\gamma_{i} := - \frac{\mu}{\lambda} f_{i}(A^{*}_{i})
$
where $A^{*}_{i}$ is the schedule of machine $i$ (at the end of the instance).

\begin{lemma}
The defined variables form a dual feasible solution. 
\end{lemma}
\begin{proof}
The first dual constraint follows immediately the definitions of $\delta_{j}, \beta_{i,j,k}$
and the decision of the algorithm. 
Specifically, the right-hand side of the constraint represents $1/\lambda$ times 
the increase of energy if a job $j$ follows a strategy $s_{i,j,k}$. This is larger than 
$1/\lambda$ times the minimum increase of energy optimized over all strategies in $\mathcal{S}_{j}$, 
which is $\delta_{j}$. 

We now show that the second constraint holds. Fix a machine $i$ and an arbitrary configuration $A$ on machine $i$.
The corresponding constraint reads
\begin{align}	\label{eq:energy-nonpreemptive-smooth}
& - \frac{\mu}{\lambda} f_{i}(A^{*}_{i}) + \frac{1}{\lambda} \sum_{(i,j,k) \in A} 
							\biggl[ f_{i}(A^{*}_{i,\prec j} \cup s_{i,j,k}) - f_{i}(A^{*}_{i,\prec j}) \biggr]
\leq f_{i}(A) \Leftrightarrow	\notag \\
& \sum_{(i,j,k) \in A} \biggl[ f_{i}(A^{*}_{i,\prec j} \cup s_{i,j,k}) - f_{i}(A^{*}_{i,\prec j}) \biggr]
\leq \lambda f_{i}(A) + \mu f_{i}(A^{*}_{i})
\end{align}
We argue that this inequality follows the $(\lambda,\mu)$-smoothness of energy power functions.
We slightly abuse notation by defining $A^{*}_{i,\prec j}(t)$ as the speed of machine $i$ (due to the 
algorithm) at time $t$ before the arrival of job $j$ and $s_{i,j,k}(t)$ be the speed at time $t$ of job $j$
if it follows the strategy $s_{i,j,k}$. Observe that $A^{*}_{i,\prec j}(t)$ is the sum of speeds (according to the algorithm) 
at time $t$ of jobs assigned to machine $i$ prior to job $j$.
For any time $t$, as the power $P_{i}$ is $(\lambda,\mu)$-smooth,
we have
\begin{align*}
 \sum_{(i,j,k) \in A} \biggl[ P_{i} \bigl(A^{*}_{i,\prec j}(t) + s_{i,j,k}(t) \bigr) - P_{i}\bigl( A^{*}_{i,\prec j}(t) \bigr) \biggr]\\
\leq \lambda P_{i}\biggl( \sum_{(i,j,k) \in A} s_{i,j,k}(t) \biggr) + \mu P_{i}\bigl(A^{*}_{i}(t)\bigr)
\end{align*}
Summing over all times $t$, Inequality~(\ref{eq:energy-nonpreemptive-smooth}) holds.
Therefore, the lemma follows.
\end{proof}

We are now ready to prove the Theorem~\ref{thm:energy}.

\begin{proof}[Proof of Theorem~\ref{thm:energy}]
By the definition of the dual variables, the dual objective is 
\begin{align*}
\sum_{j} \delta_{j} + \sum_{i} \gamma_{i} = 
\sum_{i} \frac{1}{\lambda} f_{i}(A^{*}_{i})  - \sum_{i} \frac{\mu}{\lambda}  f_{i}(A^{*}_{i})
= \frac{1-\mu}{\lambda} \sum_{i}  f_{i}(A^{*}_{i})
\end{align*}
Besides, the cost of the solution due to the algorithm is $\sum_{i} f_{i}(A^{*}_{i})$.
Hence, the competitive ratio is at most $\lambda/(1-\mu)$. 
In particular, the power functions of the form $P_{i}(s) = s^{\alpha_{i}}$, $\alpha_{i} > 1$,
are $O\bigl(\alpha^{\alpha - 1},\frac{\alpha - 1}{\alpha}\bigr)$-smooth where $\alpha = \max_{i} \alpha_{i}$. Specifically, by the smooth inequalities in \cite{CohenDurr12:Smooth-inequalities}, for any sequences of non-negative real numbers 
$\{a_{1}, a_{2}, \ldots, a_{n}\}$ and $\{b_{1}, b_{2}, \ldots, b_{n}\}$
and for any $\alpha \geq 1$, 
it holds that
\begin{align*}
\sum_{i=1}^{n} \left[ \biggl( b_{i}  + \sum_{j=1}^{i} a_{j} \biggr)^\alpha - \biggl( \sum_{j=1}^{i} a_{j} \biggr)^\alpha \right]
\leq \lambda(\alpha) \cdot \biggl( \sum_{i=1}^{n} b_{i}  \biggr)^\alpha + 
	\mu(\alpha) \cdot \biggl( \sum_{i=1}^{n} a_{i}  \biggr)^\alpha 
\end{align*}
where $\mu(\alpha) = \frac{\alpha-1}{\alpha}$ and $\lambda(\alpha) = \Theta\left(\alpha^{\alpha-1}\right)$.  

That implies the competitive ratio $O\bigl(\alpha^{\alpha}\bigr)$.
\end{proof}
\section{Conclusions}
\label{sec:conclusions}

This paper considered designing online non-preemptive  schedulers
--- a domain which has long resisted  algorithms with strong worst case guarantees . 
The paper gave provably competitive algorithms in the rejection model.  This shows how relaxed models can give rise to good algorithms for the non-preemptive setting. 
It is of significant interest to develop other realistic relaxations of worst case models (like rejection or resource augmentation) that give rise to strong algorithms for non-preemptive settings.

\bibliographystyle{unsrt}
\bibliography{references}




\end{document}